\documentclass{amsart}
\usepackage{bbm}
\usepackage[utf8]{inputenc}

\def\<#1>{\langle#1\rangle}
\def\Li{\operatorname{Li}}
\let\ideal\unlhd
\def\Q{{\mathbb Q}}
\def\N{{\mathbb N}}
\def\ann{\operatorname{ann}}
\def\lcm{\operatorname{lcm}}

\def\bd{{\boldsymbol{\partial}}}
\let\bpartial\bd
\def\bm{{\boldsymbol{m}}}
\def\bx{{\boldsymbol{x}}}
\def\bt{{\boldsymbol{t}}}
\def\bu{{\boldsymbol{u}}}
\def\bdx{{\boldsymbol{\partial_x}}}
\def\bdt{{\boldsymbol{\partial_t}}}

\def\balpha{\boldsymbol{\alpha}}
\def\bbeta{\boldsymbol{\beta}}
\def\bgamma{\boldsymbol{\gamma}}
\def\bS{{\boldsymbol{S}}}   
\def\expo{\operatorname{E}}
\def\Ore{{\mathbb O}}

\newtheorem{theorem}{Theorem}
\newtheorem{proposition}{Proposition}
\newtheorem{lemma}{Lemma}
\theoremstyle{definition} 
\newtheorem{definition}{Definition}
\newtheorem{example}{Example}

\def\parag#1{\subsubsection*{#1}}

\usepackage[plainpages=false,pdfpagelabels,colorlinks=true,citecolor=blue,hypertexnames=false]{hyperref}  

\hypersetup{
pdftitle = {A Non-Holonomic Systems Approach to Special Function Identities},
pdfauthor = {Fr\'ed\'eric Chyzak, Manuel Kauers, Bruno Salvy},
pdflang = {en}
}

\begin{document}

 \title{A Non-Holonomic Systems Approach to Special Function Identities}

\author{Fr\'ed\'eric Chyzak}
\email{frederic.chyzak@inria.fr}
\author{Manuel Kauers}
\email{mkauers@risc.jku.at}
\author{Bruno Salvy}
\email{Bruno.Salvy@inria.fr}
\thanks{F.\,C.\ and B.\,S.'s work was supported in part by the Microsoft Research\,--\,INRIA Joint Center. This work was carried out while M.\,K.\ was member of the Algorithms Project-Team at INRIA Paris-Rocquencourt. At RISC, M.\,K.\ is supported by FWF grant P19462-N18.}

\begin{abstract}
  We extend Zeilberger's approach to special function identities to cases that are not
  holonomic. The method of creative telescoping is thus applied to definite sums or
  integrals involving Stirling or Bernoulli numbers, incomplete Gamma function or
  polylogarithms, which are not covered by the holonomic framework.  The basic idea is to take
  into account the dimension of appropriate ideals in Ore algebras. This unifies several
  earlier extensions and provides algorithms for summation and
  integration in classes that had not been accessible to computer algebra before.
\end{abstract}
\maketitle

\section{Introduction}
In a classical article entitled ``A holonomic systems approach to special functions identities''~\cite{Zeilberger1990}, Doron~Zeilberger has shown that the theory of holonomic D-modules leads to algorithms for proving identities in large classes of special functions. 
In this setting, 
a function $f(x_1,\dots,x_n)$ is represented by a system of linear differential equations with polynomial coefficients that annihilate it. The function is ``holonomic''  when it possesses two important properties:
\emph{(i)\/}~besides its defining system, $f$ can be specified by a \emph{finite\/} number of initial conditions;
\emph{(ii)\/}~the number of linearly independent functions among all $x_1^{m_1}\dotsm x_n^{m_n}\frac{\partial^{k_1}}{\partial x_1^{k_1}}\dotsm\frac{\partial^{k_n}}{\partial x_n^{k_n}}(f)$ with $m_1+\dots+m_n+k_1+\dots+k_n\le N$ grows like~$O(N^{n})$.
The first property has the consequence that many operations on holonomic functions reduce to linear algebra. It leads to closure properties under sum, product, and specialization. The second one is related to the notion of holonomic D-modules and opens the way to algorithms for definite integration and summation. For this, Zeilberger developed a general method called \emph{creative telescoping}, for which he gave two algorithms: one for the general holonomic case and a faster one in the hypergeometric case.

Originally, the notion of a holonomic system is only defined for differential systems, but there are several ways
of extending it to systems of difference (or $q$-difference) equations. Among those, we prefer the use of Ore algebras~\cite{ChyzakSalvy1998}. There, the first property above corresponds to zero-di\-men\-sion\-al ideals, which are called $\partial$-finite.
The same closure properties (sum, product, specialization) hold and can be performed by Gröbner bases computation. Chyzak~\cite{Chyzak2000} extended Zeilberger's fast hypergeometric creative telescoping to all $\partial$-finite ideals,  termination being guaranteed inside Zeilberger's holonomic class.

Another direction of extension concerns functions or sequences that
cannot be defined by a holonomic system or even a $\partial$-finite ideal. 
Majewicz~\cite{Majewicz1996} has given an algorithm 
that is able to produce Abel's summation identity 
\[
 \sum_{k=0}^n\binom nk i(k+i)^{k-1}(n-k+j)^{n-k}=(n+i+j)^n
\]
automatically and to find similar new identities.
Kauers~\cite{Kauers2007} has given a summation algorithm applicable to sums involving Stirling numbers
and similar sequences defined by triangular recurrence equations. This algorithm finds, for instance, 
\[
 \sum_{k=0}^n(-1)^{m-k}k!\binom{n-k}{m-k}S_2(n+1,k+1)=E_1(n,m),
\]
where $S_2$ and $E_1$ refer to the Stirling numbers of second kind and the Eulerian numbers of first kind, 
respectively.
A summation algorithm of Chen and Sun~\cite{ChenSun2008} is able to discover certain summation identities involving 
Bernoulli numbers~$B_n$ or similar quantities, for example 
\[
 \sum_{k=0}^m\binom mk B_{n+k}=(-1)^{m+n}\sum_{k=0}^n\binom nk B_{m+k}.
\]
None of the quantities covered by these algorithms admits a definition via 
a $\partial$-finite ideal, but
all three algorithms are based on principles that resemble those employed for holonomic systems and $\partial$-finite ideals.
In each case, it turns out that the differential/difference equations defining the integrand/summand are of a 
form that permits to prove the existence of at least one non-trivial differential/difference equation for the 
integral/sum by a counting argument. 

In this article, we give algorithms dealing with ideals of Ore algebras that are \emph{not\/} $\partial$-finite. They generalize the algorithms known for the $\partial$-finite case and cover the extensions to non-holonomic functions discussed so far. Holonomy being lost, it is not always the case that creative telescoping can succeed. However, holonomy is only a sufficient condition. We show that by considering more generally the \emph{dimension\/} of the ideals and another quantity that we call \emph{polynomial growth\/}, it is possible to predict termination of a generalization of Chyzak's generalization of Zeilberger's fast algorithm. As special cases, we recover
holonomic systems (dimension~0, polynomial growth~1), but also the special purpose algorithms mentioned above for Abel-type sums (dimension~2, polynomial growth~1),
Stirling-number identities  and summation identities about Bernoulli numbers (dimension~1, polynomial growth~1). 
In addition, we get for free a summation/integration algorithm that can deal with non-holonomic special functions 
such as the incomplete Gamma function~$\Gamma(n,z)$, the Hurwitz zeta function~$\zeta(n,z)$, polylogarithms~$\Li_n(x)$,~\@\dots\
Examples are given in Section~\ref{sec:examples}.

\section{Ore Algebras and their Ideals}
Our motivation for using Ore algebras is the convenient polynomial
representation of linear operators that they offer. Classical notions of commutative polynomial rings generalize to this setting. In this section, we recall without proof the basic definitions and facts we use (see~\cite{ChyzakSalvy1998,Kandri-RodyWeispfenning1990} 
and their references for proofs, details, and history).

\begin{table}
\begin{center}\small
\begin{tabular}{lll}
\hline\hline
Operator               &$\partial\cdot f(x)$\hspace{15mm}
&$x\cdot f(x)$\hspace{15mm}\\
\hline
Differentiation   $\frac d{dx}$&$f'(x)$      &$xf(x)$   \\
Shift   $S$                    &$f(x+1)$      &$xf(x)$   \\
Difference   $\Delta$          &$f(x+1)-f(x)$  &$xf(x)$  \\
$q$-Dilation    $Q$            &$f(qx)$       &$xf(x)$   \\
Continuous $q$-difference      &$f(qx)-f(x)$   &$xf(x)$  \\
$q$-Differentiation $\delta^{(q)}$     &$\frac{f(qx)-f(x)}{(q-1)x}$  &$xf(x)$    \\
$q$-Shift     $S^{(q)}$            &$f(x+1)$               &$q^xf(x)$       \\
Discrete $q$-difference  $\Delta^{(q)}$ &$f(x+1)-f(x)$     &$q^xf(x)$       \\
Eulerian operator   $\Theta$      &$xf'(x)$  &$xf(x)$     \\
Mahlerian operator   $M$    &$f(x^b)$       &$xf(x)$      \\
Divided differences      &$\frac{f(x)-f(a)}{x-a}$&$xf(x)$ \\
\hline\hline
\end{tabular}
\end{center}
\caption{Some common Ore operators}
\label{Ore-op-table}
\end{table}

\subsection{Ore Algebras}
\parag{$\sigma$-derivations} 
Let $A$ be a commutative algebra over a field~$k$, and $\sigma$ an injective algebra
endomorphism of~$A$ that induces the identity on~$k$.  A $k$-linear endomorphism~$\delta$ of~$A$ is called a
\emph{$\sigma$-derivation\/} if it satisfies the skew Leibniz rule
$\delta(uv)=\sigma(u)\delta(v)+\delta(u)v$ for all $u$ and~$v$ in~$A$.

\parag{Skew-polynomial rings}
The associative ring generated over~$A$ by a new
indeterminate~$\partial$ and the relations $\partial
u=\sigma(u)\partial+\delta(u)$ for all $u\in A$ is called a
\emph{(left) skew-polynomial ring}.  It is denoted
by $A[\partial;\sigma,\delta]$. It does not have zero-divisors. It is called an \emph{extension} of the skew-polynomial ring $B[\partial;\sigma',\delta']$ when~$B\subset A$, $\sigma'=\sigma|_B$ and $\delta'=\delta|_B$. In examples, we allow ourselves to use different symbols in place
of~$\partial$ for increased readability. In particular, we use $S_n,S_k$, etc.\ for denoting both the indeterminate $\partial$
and the corresponding $\sigma$ in shift algebras. 

\parag{Ore operators}
A skew polynomial~$L\in A[\partial;\sigma,\delta]$ acts on  left $A[\partial;\sigma,\delta]$-modules---in most of
our applications, these are modules of functions, power series, or sequences. Solving then means finding an element~$h$ in such a module such that~$L\cdot h=0$. In this perspective, skew polynomials are called \emph{Ore operators}. 

For any~$\lambda\in A$, the action $\partial:a\mapsto \partial\cdot a=\lambda\sigma(a)+\delta(a)$ turns the algebra~$A$ itself into a left $A[\partial;\sigma,\delta]$-module. For any~$u,v$ in~$A$, one has the product rule $\partial\cdot uv=\sigma(u)\partial\cdot v+\delta(u)v$.

Tables \ref{Ore-op-table} and~\ref{Ore-alg-table} illustrate some
common types of Ore operators when~$A=k[x]$, together with the values of
$\sigma$~and~$\delta$ that define the associated skew-polynomial
ring.

In all these examples, $\sigma$~and~$\delta$ can be written as $A$-linear combinations of~$\partial$ and the identity. We call \emph{linear\/} a skew-polynomial ring with this property.

\begin{table}
\begin{center}\small
\begin{tabular}{l@{\hskip2mm}c@{\hskip3mm}c@{\hskip3mm}c@{}c}
\hline\hline
Operator                &$\sigma$ &$\delta$
              &$\partial x$\\
\hline
Differentiation         &$\operatorname{Id}$&$\frac d{dx}$
        &$x\partial+1$\\
Shift $S$                  & $S$       &0
            &$(x+1)\partial$\\
Difference              &$S$       &$\Delta$
          &$(x+1)\partial+1$\\
$q$-Dilation            &$Q$        &0
             &$qx\partial$\\
Cont. $q$-difference    &$Q$        &$Q-\operatorname{Id}$
            &$qx\partial+(q-1)x$\\
$q$-Differentiation     &$Q$        &\kern-1ex$\tfrac1{(q-1)x}(Q-\operatorname{Id})$\kern-1em\null
              &$qx\partial+1$\\
$q$-Shift               &$Q$        &0
                &$qx\partial$\\
Discr. $q$-difference   &$Q$        &$Q-\operatorname{Id}$
           &$qx\partial+(q-1)x$\\
Eulerian operator       &$\operatorname{Id}$      &$x\frac d{dx}$
         &$x\partial+x$\\
Mahlerian operator      &$M$       &0
          &$x^b\partial$\\
Divided differences     &\kern-1ex$f\mapsto f(a)$\kern-1ex\null         &\kern-1ex$x\mapsto\frac{f(x)-f(a)}{x-a}$\kern-1em\null
           &$a\partial+1$\\
\hline\hline
\end{tabular}
\end{center}
\caption{Corresponding skew-polynomial rings and their commutation rules}
\label{Ore-alg-table}
\end{table}

\parag{Ore algebras}
Ore algebras are a generalization of skew-polynomial rings well suited to the manipulation of multivariate special functions.
Let $C(\bx)=C(x_1,\dots,x_m)$ be a field of characteristic~0
of rational functions, $(\sigma_1,\dots,\sigma_n)$ be $n$ $C$-algebra morphisms of $C(\bx)$ commuting pairwise, and for each~$i$, let $\delta_i$ be a $\sigma_i$-derivation, such that the $\delta_i$'s commute pairwise and commute with the~$\sigma_j$'s when $j\neq i$. The \emph{Ore algebra\/}~$\Ore_\bx=C(\bx)\langle\bd\rangle$ is the associative $C(\bx)$-algebra generated by indeterminates $\bd=\{\partial_1,\dots,\partial_n\}$ modulo the relations
\begin{equation}\label{Ore-commutation}
 \partial_ia=\sigma_i(a)\partial_i+\delta_i(a)\quad(a\in C(\bx)),\qquad
  \partial_i\partial_j=\partial_j\partial_i.
\end{equation}

An Ore algebra~$\Ore$ or an extension $A\otimes_{C(\bx)}\Ore$ of it is called \emph{linear\/} when for each~$i$, both~$\sigma_i$ and~$\delta_i$ can be expressed as $C(\bx)$-linear combinations of $\partial_i$ and the identity.

Given two Ore algebras $\Ore_{\bx}$ and~$\Ore_{\bt}$, we write $\Ore_{\bx,\bt}$ or with an abuse of notation $C(\bx,\bt)\langle\bdx,\bdt\rangle$ for $\Ore_{\bx} \otimes_C \Ore_{\bt}$.

\subsection{Ideals}
We write $I\ideal R$ to denote that $I$~is a left ideal
in the ring~$R$, and $I=\<p_1,\dots,p_k>$ to denote that $I$ is generated by
$p_1,\dots,p_k$. 
The study of a function or sequence~$f$ translates algebraically into the study of its \emph{annihilating\/} ideal~$\ann_Af$, i.e., the left ideal of operators in an appropriate algebra~$A$ that annihilate~$f$.
(We also write~$\ann f$ when no ambiguity on~$A$ can arise.)
Computations concern finding generators of this ideal or, at least, of a sufficiently large subideal of it. This is in particular the case for creative telescoping that we study here. It computes an ideal annihilating the definite integral or sum of interest starting from a description of an ideal annihilating the summand or integrand.
 
\subsection{Gröbner Bases}
\parag{Terms}
If $f=\sum{c_{\balpha}\bd^{\balpha}}$ is a polynomial, each $c_{\balpha}\bd^{\balpha}$ for which $c_{\balpha}\neq0$ is called a \emph{term\/} of~$f$, $c_{\balpha}$ is its coefficient, $\bd^{\balpha}$ its monomial, and $\balpha$ its exponent. The total degree of~$f$ is the maximum~$|\balpha|$ over its terms, where we use the notation\\ $|\balpha|=|(\alpha_1,\dots,\alpha_n)|=\alpha_1+\dots+\alpha_n$. We also use~$|S|$ for the cardinality of a set~$S$, but this should not create confusion.

\parag{Monomial orders} A \emph{monomial order\/} is a total order on the monomials that is compatible with the product and does not have infinite descending chains. 
A \emph{graded order\/} is a monomial order such that ${\bd}^{\balpha}>{\bd}^{\bbeta}$ whenever $|\balpha|>|\bbeta|$.

\parag{Gröbner bases}
For a given monomial order, the leading term of a polynomial~$f$ is the term with largest monomial for that order. We write~$\expo(f)$ for its exponent in~$\N^n$. The crucial property that lets the theory parallel that of the commutative case is that~$\expo(fg)=\expo(f)+\expo(g)$. 
The set~$S_I:=\{\expo(f),f\in I\}$ thus has the property
$S_I=S_I+\N^n$.
The complement of~$S_I$ is a finite union of translates of coordinate subspaces.
A \emph{Gröbner basis~$G$} of~$I$ is a set of generators~$I$ such that~$S_I$ is the sum of~$\expo(g)+\N^{n}$ over $g\in G$.

\subsection{Hilbert Dimension}
We write $R^{(s)}$ for the set of polynomials of an Ore algebra~$R$ of total degree at most~$s$.
Then $I^{(s)}:=I\cap R^{(s)}$ is a vector space over $C(\bx)$. As in the commutative case, the Hilbert function of the ideal~$I$ is defined by~$\operatorname{HF}_I(s):=\dim(R^{(s)}/I^{(s)})=\dim R^{(s)}-\dim I^{(s)}$; for~$s$ large enough this function is equal to a polynomial whose degree is called the \emph{(Hilbert) dimension\/} of the ideal. We denote this integer~$\dim I$, or $\dim_R I$ when we want to make the ambient ring explicit. A reference for the results of this section is~\cite{KondratievaLevinMikhalev1999}.

\begin{example}[$\partial$-finite ideals]
An ideal is  \emph{$\partial$-finite} if its dimension is~0. This special class of ideals has been studied a lot from the computational point of view: in this case, the quotient~$R/I$ is a finite-dimensional vector space, so that many techniques of linear algebra apply.
\end{example}

\begin{example}[Hypergeometric terms]\label{ex:hgm} An $n$-variate sequence
$u_{m_1,\dots,m_n}$ is a \emph{hypergeometric term}
if
	\[ \frac{u_{m_1,\dots,m_{i-1},m_i+1,m_{i+1},\dots,m_n}}{u_{m_1,\dots,m_n}}\in\Q(m_1,\dots,m_n)\]
for~$i=1,\dots,n$. If $\bS=(S_1,\dots,S_n )$, $S_i$ representing the \emph{shift operator}, the annihilating ideal of such a sequence in the algebra $\Q(\bm)\langle\bS \rangle$ contains operators~$S_1-r_1(\bm),\dots,S_n-r_n(\bm)$. It is therefore $\partial$-finite and moreover, the dimension of the quotient as a vector space is~1.
\end{example}

\begin{example}[Stirling numbers]
	Stirling numbers of the second kind, $S_2(n,k)$, satisfy the linear recurrence
	\[S_2(n,k)=S_2(n-1,k-1)+kS_2(n-1,k).\]
The ideal generated by this relation in $\Q(n,k)\langle S_n,S_k\rangle$ has dimension~1. 
Properties of the generating series imply that this is not a $\partial$-finite sequence, so that 1~is the lowest possible dimension to work with.
\end{example}

The dimension of an ideal can be computed from a Gröbner basis for a graded order: it is the largest dimension of coordinate subspaces of~$\N^{n}$ that belong to the complement of the set of exponents of the leading terms~\cite[p.~449]{BeckerWeispfenning1993}.

We note the following inequalities that also hold in this non-commutative context:
\begin{align*}
I\subset J&\Rightarrow \dim J\le\dim I,\\
I\ideal C(\bx,\bt)\langle\bdx,\bdt\rangle &\Rightarrow \dim_{C(\bx,\bt)\langle\bdx\rangle}(I\cap C(\bx,\bt)\langle\bdx\rangle)\notag\\
   &\qquad\le\dim_{C(\bx,\bt)\langle\bdx,\bdt\rangle} I.
\end{align*}      
The first one follows from the inclusion of the vector spaces $I^{(s)}\subset J^{(s)}$. The second inequality involves dimensions relative to two different ambient rings. It can be seen by considering the following vector spaces: $F=C(\bx,\bt)\langle\bdx,\bdt\rangle^{(s)}$ contains $G=C(\bx,\bt)\langle\bdx\rangle^{(s)}$ and $H=I^{(s)}$. Then the inequality follows from $G/(G\cap H)\subset F/H$.

We also make use of the following properties of dimension.
\begin{lemma} Let $I\ideal \Ore_\bx=C(\bx)\langle\bdx\rangle$ with $\dim I=d$. 
\begin{enumerate}
	\item For any $\bdt\subset\bdx$,
	$|\bdt|\ge d+1\Longrightarrow I\cap C(\bx)\langle\bdt\rangle\neq\{0\}.$
	\item There exists~$\bdt\subset\bdx$ of cardinality~$d$ such that\\ $I\cap C(\bx)\langle\bdt\rangle=\{0\}$. 
\end{enumerate}
\end{lemma}
\begin{proof}             
We sketch the proof which shows how the notions presented so far interact (exactly as in the commutative case). 
By definition of the Hilbert function,
\begin{align*}
&\dim I^{(s)}+\dim C(\bx)\langle\bdt\rangle^{(s)}\\
={}&\dim \Ore_\bx^{(s)}+\dim C(\bx)\langle\bdt\rangle^{(s)}-\operatorname{HF}_I(s)\\
={}&\dim \Ore_\bx^{(s)}+\binom{|\bdt|+s}{s}-\operatorname{HF}_I(s).
\end{align*}
The binomial is a polynomial in~$s$ of degree $|\bdt|\ge d+1$ with positive leading coefficient, so that for large enough~$s$, it is larger than~$\operatorname{HF}_I(s)$. The sum of the dimensions of the vector subspaces on the left-hand side is therefore larger than the dimension of~$\Ore_\bx^{(s)}$ and thus they intersect nontrivially.

The second part follows by considering the exponents of leading
terms. Now the combinatorial theory is exactly as
in~\cite[Ch.~9]{CoxLittleOShea1996}.  There exists a coordinate
subspace of~$\N^{n}$ of dimension~$d$ in the complement set. This
means that there exists a subset~$\bd_{\bt}\subset\bd_{\bx}$ such that
no monomial in~$\bd_{\bt}$ is a leading term of an element
of~$I$. Thus there cannot be an element of~$I$ in those variables
only, as was to be proved.
\end{proof}

\section{Closure Properties} \label{closure}    
Our main result, Thm.~\ref{thm:ct} in the next section, generalizes the fact that holonomy is preserved under definite integration. First, we show how addition, multiplication, and action by~$\partial$ behave with respect to dimension, generalizing the corresponding closure properties for $\partial$-finite ideals.

\begin{theorem}[Closure Properties]\label{thm:closure}
 Let $I_1,I_2\ideal\Ore_\bx$ and let $f_1,f_2$ be annihilated by $I_1,I_2$, respectively.
 Then:
 \begin{enumerate}
 \item\label{item:1} $\dim\ann(\partial\cdot f_1)\leq\dim I_1$ for all~$\partial$ in~$\{\bdx\}$.
 \item\label{item:2} $\dim\ann(f_1+f_2)\leq\max(\dim I_1,\dim I_2)$.
 \item\label{item:3} If $f_1,f_2$ belong to the coefficient ring of a linear extension of\/~$\Ore_\bx$, then
$\dim\ann(f_1f_2)\leq\dim I_1+\dim I_2$.
 \end{enumerate}
\end{theorem}
\begin{proof}
We show part~\ref{item:3}. The arguments for parts \ref{item:1} and~\ref{item:2} are similar and simpler.

Setting $k:=\dim I_1+\dim I_2+1$, it suffices to show that any $k$~elements
$\partial_1,\dots,\partial_k$ among $\bpartial$ are dependent modulo~$\ann(f_1f_2)$. 

Given a Gröbner basis for $I_1$ with respect to a graded order, each polynomial~$P\in\Ore_\bx$ can be reduced to a normal form~$\overline{P}$ such that~$P\cdot f_1=\overline{P}\cdot f_1$ and $\deg\overline{P}\le\deg P$. Moreover, by definition of the dimension, the set of all monomials in~$\{\overline{P}\mid P\in R^{(s)}\}$ has cardinality growing like~$O(s^{\dim I_1})$. The same considerations hold for~$f_2$.

By induction on the degree of~$P$, the condition of a linear extension implies that $P\cdot(f_1f_2)$ rewrites as a linear combination of  monomials~$(\bpartial^{\bbeta}\cdot f_1)(\bpartial^{\bgamma}\cdot f_2)$, with~$|\bbeta+\bgamma|\le\deg P$. Moreover, we can assume that the monomials have been reduced to their normal forms and the inequalities still hold.

Let~$s\ge0$ and consider the following identities
\begin{equation}\label{eq:3}
 \partial_1^{\alpha_1}\cdots\partial_k^{\alpha_k}\cdot f_1f_2= 
  \sum_{|\bbeta|\leq s}\sum_{|\bgamma|\leq s}
    u_{\balpha;\bbeta,\bgamma}(\bx)
    \bigl(\bpartial^{\bbeta}\cdot f_1\bigr)
    \bigl(\bpartial^{\bgamma}\cdot f_2\bigr),
\end{equation}
($|\balpha|\leq s$), where the sums on the right are constructed as above.
Then the first sum actually ranges over a subset $\{\bbeta,|\bbeta|\leq s\}$ of cardinality
$O(s^{\dim I_1})$ and the second one ranges over a subset $\{\bgamma,|\bgamma|\leq s\}$ of 
cardinality $O(s^{\dim I_2})$. Thus there is a generating set of $O(s^{\dim I_1}s^{\dim I_2})
= O(s^{k-1})$ monomials for all the summands. This implies that for $s$ large
enough, there exists a nontrivial linear combination of the $O(s^k)$ polynomials in~\eqref{eq:3}
of the form
\[
 \sum_{|\balpha|\leq s} w_{\balpha}(\bx)\partial_1^{\alpha_1}\cdots\partial_k^{\alpha_k}\cdot(f_1f_2)=0,
\]
as we wanted to show.  
\end{proof}

\parag{Algorithm}
The proof gives an algorithm that computes generators of a subideal of the desired annihilating ideal, with a dimension that obeys the inequality. For increasing~$s$, compute the normal forms of all monomials in~$\bpartial$ of degree at most~$s$, compute linear combinations between them (the kernel of the matrix $(u_{\balpha;\bbeta,\bgamma})$), and return these relations if they are sufficiently many to obtain the dimension of the theorem.  For $\partial$-finite ideals, this returns  the same result as the algorithms in~\cite{ChyzakSalvy1998}. Various optimizations are possible. The other closure operations are similar.

\begin{example}\label{ex:doublestirling}
The sequence
\[
 f_{n,m,k,l}=\binom nkS_2(k,l)S_2(n-k,m),
\]
is annihilated by an ideal of dimension at most~2. 
This follows from Thm.~\ref{thm:closure} by observing: \emph{(i)\/}~that $\binom nk$ is hypergeometric, thus
$\partial$-finite, and thus annihilated by an ideal of dimension~0; \emph{(ii)\/}~that $S_2(k,l)$ and $S_2(n-k,m)$ are
Stirling-like (see Section \ref{Stirling-like}) and thus annihilated by certain ideals of dimension~1. 
More specifically, the factors $\binom nk$, $S_2(k,l)$, and $S_2(n-k,m)$ are annihilated by the ideals
\begin{alignat*}1
 &\<(k-n-1) S_n+n+1,(k+1) S_k+k-n,S_m-1,S_l-1>,\\
 &\<S_n-1,S_k S_l-(l+1) S_l-1,S_m-1>,\\
 &\<S_n S_k-1,(m+1) S_m S_k+S_k-S_m,S_l-1>,
\end{alignat*}
respectively. The algorithm sketched above yields
\begin{alignat*}1
 I:=\langle &1 + n + (1 + m) (1 + n) S_m - (1 - k + n) S_nS_m, \\
   &(k - n) S_m + 
 (1 + k) S_kS_l  + (1 + k) (1 + m) S_kS_lS_m\\
 &{}+ (1 + l) (k - n) S_lS_m ,
 1 + n + (1 + l) (1 + n) S_l\\
 &{} - (1 + k) S_kS_lS_n\rangle
	\ideal C(n,m,k,l)\<S_n,S_m,S_k,S_l>
\end{alignat*}
as an ideal of annihilators of $f_{n,m,k,l}$.
It has dimension~2.
\end{example}

\section{Creative Telescoping}
Creative telescoping is basically a combination of differentiation under the integral sign and integration by parts, or analogues for other operators. We now give it an algebraic interpretation. Our main theorem can be viewed as predicting cases when identities are bound to exist.

\subsection{Telescoping of an Ideal}

The heart of the method of creative telescoping translates algebraically into the notion of the telescoping of an ideal.
\begin{definition}
Let $I\ideal \Ore_{\bx,\bt}=C(\bx,\bt)\langle\bdx,\bdt\rangle$ be a left ideal.
Assume $|\bt|=|\bdt|$.
We define the \emph{telescoping} of~$I$ with respect to~$\bt=(t_1,\dots,t_k)$ as the left ideal of $C(\bx)\langle\bdx\rangle$
\[
 T_{\bt}(I):=\left(I+\partial_{t_1}\Ore_{\bx,\bt}+\dots+\partial_{t_k}\Ore_{\bx,\bt}\right)\cap C(\bx)\langle\bdx\rangle.
\]
\end{definition}

\begin{definition}
The variables~$\bdt=(\partial_{t_1},\dots,\partial_{t_k})$ of the Ore algebra~$\Ore_{\bx,\bt}$ are \emph{telescopable} if there exist elements $a_1,\dots,a_k$ in $C(\bx,\bt)$ 
such that
\[\delta_{t_i}(a_i)\in C(\bx)\setminus\{0\}\quad\text{and}\quad\sigma_{t_i}(a_i)\partial_{t_j}=\partial_{t_j}\sigma_{t_i}(a_i)\ \ (j\neq i).\]
\end{definition}
Note that this is a condition on the algebra, and does not depend on any specific ideal. 
In view of Table~\ref{Ore-op-table}, this is not a strong restriction for our applications.
For example, the differential operator~$d/dt$ and the difference operator~$\Delta_t$ are telescopable, with $a=t$, but the shift operator is not.
This notion lets us generalize an idea of Wegschaider~\cite{wegschaider1997} used at the end of the proof of Thm.~\ref{thm:ct} below.

\subsection{Polynomial Growth}

In Thm.~\ref{thm:ct} below we give an upper bound for the dimension of~$T_\bt(I)$, thus providing a termination criterion for the algorithms in Section~\ref{algos}.
Our bound depends on the dimension of~$I$ and on its ``polynomial growth'', defined as follows.

\begin{definition}\label{def:p}
 The left ideal~$I\ideal \Ore_{\bx,\bt}$ has \emph{polynomial growth} $p$ with respect to a given graded order if there exists a sequence of polynomials~$P_s(\bx,\bt)$, $s\in\N$,
 such that for any $\balpha$ with
 $|\balpha|\leq s$, the normal form of $P_s(\bx,\bt)\bd^\alpha$ 
 with respect to a Gröbner basis of~$I$ for the order
 has coefficients in $C(\bx)[\bt]$
 whose degrees with respect to $\bt$ are $O(s^p)$.  
 We say that $I$ has polynomial growth~$p$ when there is a graded order with respect to which it does.
\end{definition}

If all the $\sigma$'s are automorphisms, the polynomial growth is bounded by the dimension of the Ore algebra, $|\bx|+|\bt|$.
But the interesting cases are those where the polynomial growth is 
smaller than that. 
For certain ideals of dimension~0, we get a better estimate in Thm.~\ref{thm:p-for-diff-diff-alg} below.
For an arbitrary ideal in an arbitrary Ore algebra, we do not know how to determine 
its polynomial growth algorithmically yet.

\begin{example}\label{ex:doublestirling:2}
 The basis at the end of Example~\ref{ex:doublestirling} is a Gröbner basis with
 respect to a graded order. Inspection of its leading coefficients shows that the 
 sequence
 \[
   P_s(n,m,k):=\prod_{|j|\leq s}(1+n-k+j)(1+k+j)^2
 \]
 satisfies the conditions of Definition~\ref{def:p}. Since $\deg P_s=O(s)$, it follows that
 $I$ has polynomial growth~1.
\end{example}

Just before stating our result, we now give notation and a definition for algebras amenable to it.

\begin{definition}
Set $R=C(\bx)[\bt]$ and let $R_{\leq n}$~be the set of elements~$A\in R$ such that\/~$\deg_\bt A\leq n$.
A \emph{difference-differen\-tial algebra} is an Ore algebra~$\Ore_{\bx,\bt}$  such that for any~$i$, either 
\begin{itemize}
\item $\delta_i=0$ and, for any $u\in R$, $\deg_\bt\sigma_i(u) = \deg_\bt u$;
\item or $\sigma_i=\operatorname{Id}$ and $\delta_i$~is a derivation such that, for any $u\in R$, $\deg_\bt\delta_i(u)\leq\nu+\deg_\bt u$,
\end{itemize}
$\nu\in\N$ being fixed.
We set $\mathcal S$ and~$\mathcal D$ to be the sets of~$i$ of, resp., first and second types.
\end{definition}

\begin{theorem}\label{thm:p-for-diff-diff-alg}
Let\/ $\Ore_{\bx,\bt}$ be a difference-differential algebra endowed with a graded ordering.
Let $I\ideal\Ore_{\bx,\bt}$ have dimension~0. Call~$\phi$ the map sending any $A\in\Ore_{\bx,\bt}$ to its normal form modulo~$I$ w.r.t.~the given graded ordering, and~$\Gamma$ the (finite) set of monomials in normal form.
Then, there exist $m,\ell\in\N$ and $L\in R_{\leq\ell}$ such that, for any~$i$ and any~${\bbeta\in\Gamma}$,
\begin{equation}\label{eq:uniform-reductions}
\phi(\partial_i\bd^{\bbeta}) \in \frac1L \sum_{\bgamma\in\Gamma} R_{\leq m}\bd^{\bgamma} .
\end{equation}
Define a sequence~$(P_s)_{s\in\N}$ by~$P_0=1$ and
\[P_{s+1}:=\begin{cases}
\lcm(P_s,\lcm(L\sigma_i(P_s)),P_s\lcm(L,Q_s)),&\text{if $\mathcal{D}\neq\emptyset$,}\\
\lcm(P_s,\lcm(L\sigma_i(P_s))),&\text{if $\mathcal{D}=\emptyset$,}
\end{cases}
\]
where $Q_s$~denotes the squarefree part of~$P_s$ and the lcm in~$\sigma_i$ is over all~$i$.

If $\deg_\bt P_s = \Theta(s^p)$ for some integer~$p>0$, then $I$~has polynomial growth~$p$.
\end{theorem}

Before the proof, note that, by the definition of~$\phi$, the sum in~\eqref{eq:uniform-reductions} is limited to~$|\bgamma|\leq|\bbeta|+1$.
Moreover, the existence of a uniform~$(L,m)$ is easily obtained in the zero-dimensional case, as the set of~$(i,\bbeta)$ is then finite.

\begin{proof}
Introduce the sets $F_s=\{\bd^{\balpha}, |\balpha|\leq s\}$ for~$s\in\N$.
For~$n\in\N$, define $R_{\leq n}[\Gamma]$ as $\sum_{\bgamma\in\Gamma} R_{\leq n} \bd^{\bgamma}$.
For a fixed~$s$, suppose there exists $P\in R$ and an integer~$D$ for which
\[ \phi(F_s) \subset P^{-1}R_{\leq D}[\Gamma] . \]
We look for homologues $P'$ and~$D'$ for~$F_{s+1}$.
As $\phi$~satisfies
\begin{equation*}
\phi(F_{s+1}) = \phi\Bigl(F_s\cup\bigcup_i\partial_i F_s\Bigr) \subset
  \phi(F_s) \cup \bigcup_i\phi\bigl(\partial_i\phi(F_s)\bigr) ,
\end{equation*}
we study $\phi(\partial_i A)$ for~$A=P^{-1}U \bd^{\bbeta}$ when $U\in R_{\leq D}$, $\bbeta\in\Gamma$, and $|\bbeta|\leq s$.
If $i\in\mathcal S$, then from~$\delta_i=0$ follows $\partial_i A = \sigma_i(P)^{-1}\sigma_i(U)\partial^i\bd^{\bbeta}$;
therefore,
\[ \phi(\partial_i A)\in \bigl(L\sigma_i(P)\bigr)^{-1} R_{\leq D+m}[\Gamma] . \]
Else, $i\in\mathcal D$ and from~$\sigma_i$ being the identity follows $\partial_i A = P^{-1}U\partial^i\bd^{\bbeta} + P^{-1}\delta_i(U)\bd^{\bbeta} + \delta_i(P^{-1})U\bd^{\bbeta}$;
thus, $\phi(\partial_i A)$~is in
\begin{equation*}
\frac1{LP} R_{\leq D+m}[\Gamma] +
  \frac1{LP} R_{\leq D+\ell+\nu}\bd^{\bbeta} + \frac1{PQ} R_{\leq D+D_0}\bd^{\bbeta}
\end{equation*}
where $Q$~is the square-free part of~$P$ and $D_0=\deg_\bt Q$.
Defining $P'$ as the lcm of $P$, $P \lcm(L,Q)$ if~$\mathcal D\neq\emptyset$, and the~$\sigma_i(P)L$'s for~$i\in\mathcal S$ yields $A,\phi(\partial_i A) \in (P')^{-1} R[\Gamma]$.
Next, setting $D_1=\deg_\bt(P'/PQ)$ and $D_2=\deg_\bt(P'/LP)$, then $\Delta$ as the maximum of $\max\{m,\ell\}+D_2$ and $\max\{\ell+\nu+D_2,D_0+D_1\}$ if~$\mathcal D\neq\emptyset$, finally $D'=D+\Delta$ yields
\[ \phi(F_{s+1})\subset (P')^{-1} R_{\leq D'}[\Gamma]. \]
Therefore, introducing the sequence~$(P_s)$ defined in the theorem, we get by a first induction on~$s$ that $\phi(F_s)\subset P_s^{-1} R[\Gamma]$.
Next, under the additional assumption that $\deg_\bt P_s = \Theta(s^p)$ for some integer~$p>0$, $P$ and~$P'$ in the proof above can be taken as $P_s$ and~$P_{s+1}$, so that all of $D_0$, $D_1$, and~$D_2$ are bounded by $\deg P_{s+1}-\deg P_s$, thus by~$O(s^{p-1})$.
Therefore, by another induction
\[ \phi(F_s)\subset P_s^{-1} R_{\leq O(s^p)}[\Gamma], \]
proving a polynomial growth~$p$ for~$I$.
\end{proof}

In the uses of creative telescoping for summation, the identities (or ideals) are often stated in terms of shifts, while the operation of interest is a finite difference. However, it is possible to compute the polynomial growth in a difference-differential algebra with shifts, and it will be the same when considered in an algebra with difference operators. This is the meaning of the following theorem, whose proof based on transporting graded orderings we omit.

\begin{proposition}\label{prop:1}
For a given difference-differential algebra\/ $\Ore_{\bx,\bt}$ and indeterminates 
$\bu=(u_1,\dots,u_k)$, consider
\begin{alignat*}1
 \Ore_{\bx,(\bt,\bu)}&=\Ore_{\bx,\bt}\otimes_C C(\bu)
 [\partial_{u_1};S_{u_1},0]\dots[\partial_{u_k};S_{u_k},0]
 \quad\text{and}\\
 \Ore'_{\bx,(\bt,\bu)}&=\Ore_{\bx,\bt}\otimes_C C(\bu)
 [\partial_{u_1};S_{u_1},\Delta_{u_i}]\dots[\partial_{u_k};S_{u_k},\Delta_{u_k}],
\end{alignat*} 
where $\Delta_{u_i}=S_{u_i}-\operatorname{Id}$ ($i=1,\dots,k$).
Let\/ $\mu$ be the (bijective) left-$C(\bx,\bt,\bu)$-linear map sending each~$\partial_{u_i}$ to~$\partial_{u_i}+1$ and all the $\bdx,\bdt$ to themselves.
Then an ideal~$I\ideal\Ore_{\bx,(\bt,\bu)}$ and $\mu(I)\ideal\Ore'_{\bx,(\bt,\bu)}$ have the same polynomial growth.
\end{proposition}

\subsection{Main Result}

Our main result is the following sufficient condition for creative telescoping.
For simplicity, we state it for $|\bt|=|\bdt|$.

\begin{theorem}[Creative Telescoping]\label{thm:ct}
Let~$I\ideal\Ore_{\bx,\bt}=C(\bx,\bt)\langle\bdx,\bdt\rangle$ be an ideal of dimension $d$ and polynomial growth~$p$.
If\/ $|\bt|=|\bdt|$ and the $\bdt$'s are telescopable then
\[\dim_{\Ore_\bx}T_{\bt}(I)\le d+(p-1)|\bt|,\]
whenever this bound is nonnegative. In particular, when the bound is smaller than $|\bx|$, the ideal $T_{\bt}(I)$ is non-trivial.
\end{theorem}
\begin{proof}    
If $d+(p-1)|\bt|\ge|\bx|$, there is nothing to show.
Otherwise, by Lemma~1, it is sufficient to show that 
   any $k:=d+(p-1)|\bt|+1$ elements $\partial_1,\dots,\partial_k$ in $\{\bd_\bx\}$ are dependent modulo~$T_\bt(I)$.

   Let $s\geq0$ and consider the following set of members of~$I$ from the definition of the polynomial growth of~$I$:
\begin{equation}\label{reductions}
P_s(\bx,\bt)\partial_1^{\alpha_1}\dotsm\partial_k^{\alpha_k}\bd_{\bt}^{\balpha_{\bt}}
-\sum_{|\bbeta|\le s}{c_{\balpha,\bbeta}(\bx,\bt)\bd^{\bbeta}},\quad |\balpha|\le s.
\end{equation}
The coefficients $c_{\balpha,\bbeta}(\bx,\bt)$ can be viewed as
$C(\bx)$-linear combinations of monomials in~$\bt$ of degree $O(s^p)$.
The ideal~$I$ having dimension~$d$, these sums in~\eqref{reductions} for all $\balpha$ such
that $|\balpha|\le s$ actually range over a common subset of
$\{\bbeta,|\bbeta|\le s\}$ of cardinality $O(s^d)$.  Thus, there is a
generating set of $O(s^{d+p|\bt|})$ monomials in~$\bt,\bd$ for all the
summands. This implies that for $s$ large enough, there exists a
nontrivial linear combination of the~$O(s^{k+|\bt|})$ polynomials
in~\eqref{reductions} of the form
\[
  P_s(\bx,\bt)\sum_{|\balpha|\le s}C_{\balpha}(\bx)\partial_1^{\alpha_1}\dotsm\partial_k^{\alpha_k}\bdt^{\balpha_\bt}=:P(s)Q\in I.
\]
Multiplying by~$1/P_s(\bx,\bt)$ shows that $Q\in I\cap C(\bx)\langle\bdx,\bdt\rangle$.
Now the operators~$\partial_{t_i}$ commute with the coefficients of~$Q$. Thus, $Q$ can be rewritten
\begin{equation}\label{division}
Q=R+\partial_{t_1}Q_1+\dots+\partial_{t_{|t|}}Q_{|t|},
\end{equation}
with~$R\in C(\bx)\langle\partial_1,\dots,\partial_k\rangle\in T_\bt(I)$. If~$R\neq0$, we have found the element of~$T_\bt(I)$ we were looking for. Otherwise, since~$Q\neq0$, there exists a~$Q_i\neq0$. Since~$\partial_{t_i}$ is telescopable, there is an element~$a_i\in C(\bx,\bt)$ such that 
\[\sigma_{t_i}(a_i)\partial_{t_i}=b_i+\partial_{t_i}a_i,\qquad \sigma_{t_i}(a_i)\partial_{t_j}=\partial_{t_j}\sigma_{t_i}(a_i)
\]
for some $b_i\in C(x)\setminus\{0\}$.
Multiplying~\eqref{division} by~$\sigma_{t_i}(a_i)$ yields
\[I\ni \sigma_{t_i}(a_i)Q=b_iQ_i+\partial_{t_i}\tilde{Q}_i+\dots+\partial_{t_{|t|}}\tilde{Q}_{|t|}.\]
Now, $Q_i\neq0$ can be rewritten as in~\eqref{division}. 
Repeating this process if necessary,
we eventually reach a non-zero element of~$T_\bt(I)$, as was to be proved.
\end{proof}

\subsection{Examples}\label{sec:examples}

\parag{Proper hypergeometric sequences} 
These are sequences with two indices that are hypergeometric as in Example~\ref{ex:hgm} with the further 
constraint that they can be written
\begin{alignat}1 \label{eq:properhg}
  u_{m,k}=Q(m,k)\frac{\prod_{i=1}^u{(a_im+b_ik+c_i)!}}{\prod_{i=1}^v{(u_im+v_ik+w_i)!}}\xi^k,
\end{alignat}
where~$\xi\in\mathbb{C}$, $Q$~is a polynomial and the $a_i$'s, $b_i$'s, $u_i$'s, $v_i$'s are integers.
A typical example is the binomial coefficient $\binom m k$.
Since such a sequence is hypergeometric, a Gröbner basis of $\ann u_{m,k}$ for any order is 
formed by the relations
\[
 S_m-{u_{m+1,k}/u_{m,k}}\quad\text{and}\quad S_k-u_{m,k+1}/u_{m,k}.
\]
The normal form of~$S_m^{s_1}S_k^{s_2}$ with respect to this basis is simply the rational function~$u_{m+s_1,k+s_2}/u_{m,k}$. 
Since the $a_i,b_i,u_i,v_i$'s are all integers, the common denominator of all these rational functions for~$s_1+s_2\le s$ has a degree that grows only \emph{linearly} with~$s$~\cite{PetkovsekWilfZeilberger1996}, as it is bounded
by
\begin{alignat*}1
 P_s(m,k)&=Q(m,k)\prod_{|j|\leq(|u_i|+|v_i|)s}\prod_i (u_i m+v_i k+ w_i + j)\\
         &\qquad{}\times\prod_{|j|\leq(|a_i|+|b_i|)s}\prod_i (a_i m+b_i k+ c_i + j).
\end{alignat*}
In our terminology, the corresponding annihilating ideal has polynomial growth~1, by Thm.~\ref{thm:p-for-diff-diff-alg} and Prop.~\ref{prop:1}.
Further generalizations to the multivariate and $q$-cases 
can also be considered.
Thm.~\ref{thm:ct} (with $d=0, p=1, |\bt|=1, n=2)$ generalizes the result that creative telescoping applies to proper hypergeometric sequences. 

\parag{General hypergeometric sequences}
Not every hypergeometric sequence is proper. For example, 
\[
  u_{m,k}=\frac1{mk+1}\binom{2m-2k-1}{m-1}
\]
is not. By a criterion of Abramov~\cite{Abramov2002a}, creative telescoping fails on this example, i.e., $T_{\bt}(I)=\{0\}$.
This phenomenon is well consistent with our theorem, because the nonlinear factor in the denominator implies 
that the annihilating ideal~$I$ has polynomial growth~2:
a possible choice for $P_s$ is
\[
 P_s(m,k)=\prod_{i+j\leq s}((m+i)(k+j)+1)\times\prod_{|i|\leq 3s}(m-2k+i),
\]
whose degree is quadratic in~$s$, hence the polynomial growth by Thm.~\ref{thm:p-for-diff-diff-alg} and Prop.~\ref{prop:1}.
Thm.~\ref{thm:ct} (with $d=0,p=2,|\bt|=1,n=2$) implies the trivial bound $\dim T_{\bt}(I)\leq1$, which
is reached in this example.

\parag{Holonomic Functions}
This notion was popularized for special functions in~\cite{Zeilberger1990}. 
The technical definition is related to holonomic D-modules. 
Thm.~\ref{thm:p-for-diff-diff-alg} implies that it is sufficient to consider $\partial$-finite ideals, as was  shown first by Takayama~\cite{Takayama1992}.
\begin{proposition}\label{propholo}
Let~$I\ideal\Ore_{\bx,\bt}=C(\bx,\bt)\<\bpartial_{\bx,\bt}>$ be a left ideal of dimension~0 of the 
\emph{differential} Ore algebra ($\sigma_i=\operatorname{Id}$, $\delta_i=d/dx_i$). Then $I$ has polynomial growth~1.
\end{proposition}
\begin{proof}
The crucial property is that in the differential setting, the leading coefficient of an element does not change upon multiplication by some~$\partial_i$ (see~Table~\ref{Ore-alg-table}). 
The result is a corollary of Thm.~\ref{thm:p-for-diff-diff-alg}, which is applicable since $I$~has dimension~0.
The polynomials~$P_s$ defined in the theorem here take the form~$P_s=L^s$, having degree $O(s^1)$ and thereby proving a polynomial growth~1.
\end{proof} 
Put together, Prop.~\ref{propholo} and Thm.~\ref{thm:ct} imply $\dim T_{\bt}(I)=0$, 
and thus recover the celebrated closure of holonomic functions under definite integration
as a special case.
\parag{Stirling-like sequences}    \label{Stirling-like}
This notion was introduced in~\cite{Kauers2007}. A~sequence is called Stirling-like if it is annihilated by an ideal of operators
\[
  \< u + v S_k^\alpha S_{m_1}^\beta + w S_k^\gamma S_{m_1}^\delta,
     s_2 + t_2 S_{m_2}, \dots,
     s_n + t_n S_{m_n} >,
\]
where $u,v,w,s_2,\dots,s_n,t_2,\dots,t_n$ are polynomials that split into integer-linear factors, and $\alpha,\beta,\gamma,\delta$
are integers subject to a certain nondegeneracy condition.
A typical example is $\binom mk S_2(k,m)$, where $S_2$ refers to the Stirling numbers of the second kind. 
It was shown in~\cite{Kauers2007} that such ideals lead to non-trivial creative telescoping relations whenever $n\geq2$. 
These ideals have dimension~1 and polynomial growth~1, so our theorem (with $d=1,p=1,|\bt|=1,n\geq3$) includes this result as a special case.

\parag{Abel-type sequences}
This notion was introduced in~\cite{Majewicz1996}. A sequence is called Abel-type if it can be written in the form 
\[
  u_{m,k}(k+r)^k(m-k+s)^{m-k}\frac{r}{k+r}
\]
for some proper hypergeometric term~$u_{m,k}$. The annihilating ideals of such sequences can be written 
\[
 \< a S_mS_k - b S_r, c S_m - d S_s >
\]
for certain polynomials $a,b,c,d$.
It was shown in~\cite{Majewicz1996} that such ideals lead to non-trivial creative telescoping relations.
These ideals have dimension~2, and if $u_{m,k}$ is as in~\eqref{eq:properhg}, then 
\[
 P_s(m,r,s,k)=\tilde P_s(m,k)\times\prod_{|j|\leq s}(k+r+j)(m-k+j),
\]
with $\tilde P_s(m,k)$ being the polynomial sequence stated above for proper hypergeometric terms,
justifies that their polynomial growth is~1, so our theorem (with $d=2$, $p=1$, $|\bt|=1$, $n=4$) includes this result as a special case, too. 

\parag{Bernoulli examples}
Chen and Sun~\cite{ChenSun2008} do not give a formal description of the class of summands to which their algorithm is applicable.
Most of their examples concern sums with summands of the form $h_{k,m_1,m_2} B_{a k+b m_1+c m_2}$ where $h$ is a hypergeometric term, $B$~refers to Bernoulli
numbers, Bernoulli polynomials, Euler numbers, or Euler polynomials, and $a,b,c$ are specific integers. 
The annihilating ideals of such objects have dimension~$1$, and, if $h$ is proper hypergeometric, polynomial growth~$1$.
In this case, our theorem guarantees the success of creative telescoping. 

\parag{Further examples}

Our theorem also extends to sequences and functions for which no special purpose summation
or integration algorithm has been formulated so far. 
For instance, the sequence $f_{n,m,k,l}$ considered in Examples~\ref{ex:doublestirling} and
\ref{ex:doublestirling:2} 
is not Stirling-like, but only annihilated by an ideal of dimension~2 and polynomial growth~1.
Our Theorem predicts the existence of creative telescoping relations
(cf.\ also Examples~\ref{ex:doublestirling:3} and~\ref{ex:doublestirling:4} below).
We list some further identities from the 
literature 
that were previously considered inaccessible to computer algebra, but that can be proven
with creative telescoping. 
In all the following identities, the integrand is not holonomic, but annihilated by an ideal of dimension~1
and polynomial growth~1, so our theorem predicts a priori that relations for the integral must exist.
{\allowdisplaybreaks
\begin{alignat*}1
 &\int_0^\infty x^{k-1}\zeta(n,\alpha+\beta x)\,dx=\beta^{-k}B(k,n-k)\zeta(n-k,\alpha),\\
 &\int_0^\infty x^{\alpha-1}\Li_n(-xy)\,dx=\frac{\pi(-\alpha)^ny^{-\alpha}}{\sin(\alpha\pi)},\\
 &\int_0^\infty x^{k-1}\exp(xy)\Gamma(n, xy)\,dx=\frac{\pi y^{-k}}{\sin((n+k)\pi)}\frac{\Gamma(k)}{\Gamma(1-n)},\\
 &\int_0^\pi\cos(nx-z\sin x)\,dx=\pi J_n(z). 
\end{alignat*}}%
In these identities, $B$~refers to the Beta-function, $\zeta$ to the Hurwitz zeta function, $\Li_n$ 
to the $n$th polylogarithm, $\Gamma$ to the incomplete Gamma function, and $J_n$ to the $n$th Bessel 
function. 
Greek letters refer to parameters, $n,m,k$ are discrete variables, $x,y,z$ are continuous ones.
Note that the polylogarithm, while holonomic for each specific~$n$, is not even $\partial$-finite when $n$ is
``symbolic''. Also, the integrand of the last integral, despite being elementary,
is not holonomic. 

All these identities are proven by computing operators annihilating both sides, by making use of closure properties for sums, products or~$\partial$ (Section~\ref{closure}) and by creative telescoping in the case of definite 
sums and integrals (Section~\ref{algos} below).
If the operators found in this way generate an ideal of dimension~0, then proving the identity reduces to verifying a finite number of
initial values. For higher dimensional ideals, the number of initial values to be checked may be infinite. For example, for completing
the proof of the first integral identity above, it remains to check the identity for $k=1$ and \emph{all\/} $n\geq0$:
\[
 \int_0^\infty \zeta(n,\alpha+\beta x)dx = \frac1{\beta(n-1)}\zeta(n-1,\alpha),
\]
an identity not much easier than the original one. In many instances, however, the identities to be verified as initial conditions are 
trivial. 

\section{Algorithms for Summation and Integration}\label{algos}

\subsection{A Fasenmyer-Style Algorithm}

The first algorithmic approach to symbolic summation goes back to Fasenmyer and 
was formulated for the summation of hypergeometric 
terms~\cite{PetkovsekWilfZeilberger1996}. 
The following generalized version of her algorithm is applicable to any given 
annihilating ideal in an Ore algebra. 

The algorithm follows from the proof of Thm.~\ref{thm:ct}. Given an ideal 
$I\ideal C(\bx,\bt)\<\bdx,\bdt>$, the algorithm searches for elements of 
$I\cap C(\bx)\<\bdx,\bdt>$ (called $\bt$-free operators). Proceeding by increasing
total degree, it makes an ansatz with undetermined coefficients for such 
an operator, reduces it to normal form with respect to a Gröbner basis for~$I$,
brings the normal form to a common denominator, and then compares the coefficients with respect to the $\bt$'s of the numerator to zero.
This gives a linear system of 
equations over $C(\bx)$ for the undetermined coefficients, which is then solved. 
If the system has no solution, the procedure is repeated with a larger total degree.
If a solution is found, it leads to an operator that can be written in the 
form~\eqref{division}
\[
  A+\partial_{t_1} B_1+\cdots + \partial_{t_{|t|}} B_{|t|}\in I
\]
where $0\neq A\in  T_{\bt}(I)$, by following the steps at the end of the proof that ensure~$A\neq0$.
The procedure can be repeated to search for further operators in $T_{\bt}(I)$
until all operators found generate an ideal of the dimension predicted by 
Thm.~\ref{thm:ct}.

\begin{example}\label{ex:doublestirling:3}
 For the ideal~$I$ from Example~\ref{ex:doublestirling}, with $k$ as summation variable,
 this algorithm discovers a first operator for total degree~4 that may be written in the form
\begin{multline*}
S_m + S_l + (2 + l + m)S_lS_m - S_lS_mS_n\\
{}+(S_k - 1)((m+1)S_mS_l-S_m S_n S_l+S_l).
\end{multline*}
As a consequence, we find 
 \[
   A:=S_m + S_l + (2 + l + m)S_lS_m - S_lS_mS_n\in T_k(I).
 \]
Therefore $A$ belongs to the annihilating ideal of the sum
 \[
  \sum_k \binom nkS_2(k,l)S_2(n-k,m).
 \]
 This proves that the sum is equal to $\binom{l+m}lS_2(n,l+m)$: $A$ also belongs
 to the annihilating ideal of that quantity, and it 
agrees 
 with
 the sum for $n=0$ and arbitrary $l$ and~$m$.  
\end{example}

\subsection{A Zeilberger-Style Algorithm}\label{sec:fast-algo}

 Zeilberger's ``fast algorithm'' was originally formulated for summation of proper hypergeometric terms 
 only~\cite{Zeilberger1990b}, then a differential analog was given for hyperexponential functions by Almkvist and Zeilberger~\cite{AlmkvistZeilberger1990}.
This was later extended by Chyzak~\cite{Chyzak2000} to an integration/summation algorithm for arbitrary $\partial$-finite functions, with termination guaranteed in the holonomic case.
 We extend this algorithm further to the case of integrands or summands defined by arbitrary 
 annihilating ideals. 

 The fast approach is applicable only for single sums or integrals, i.e., if $|\bt|=1$.
(Refer however to~\cite[Sec.~3.3]{Chyzak2000} for an iterated treatment of summations and integrations.)
 Let $I\ideal C(\bx,t)\<\bdx,\partial_t>$ be given by a Gr\"obner basis~$G$ with respect
 to a graded order. 
 The algorithm searches for operators
$   A + \partial_t B\in I$
 with $A\in C(\bx)\<\bdx>$ and $B\in C(\bx,t)\<\bdx,\partial_t>$. Proceeding by increasing total
 degree, it makes an ansatz for both $A$ and~$B$ and computes the normal form of
 $Q:=A+\partial_t B$ with respect to~$G$. Without loss of generality, only irreducible terms 
 need to be included in the ansatz for~$B$.
 In order for $Q$ to belong to~$I$, it is necessary and
 sufficient that all the coefficients in the normal form of $Q$ be zero. Comparing them to zero
 leads to a system of first order linear functional equations (a ``coupled system''), 
 which is then solved. While in the $\partial$-finite case this system is always square, 
 in the case of positive dimension it may be rectangular, owing to extraneous equations 
 potentially being introduced by irreducible terms in $\partial_t B$ of degree $d+1$. 
 However, it is always possible to separate the equations into a square coupled system
 (that can be solved in a first step as in Chyzak's algorithm) and a system with
 additional linear algebraic constraints (that can be accommodated in a second step by computing
 suitable linear combinations of the solutions obtained in the first step).
 Any solution $A+\partial_t B\in I$ found in this way gives rise to an element~$A\in T_t(I)$.
 The procedure may be repeated until the elements of $T_t(I)$ found in this way 
 generate an ideal in $C(\bx)\<\bdx>$ whose dimension matches the dimension predicted
 by Thm.~\ref{thm:ct}.

\begin{example}\label{ex:doublestirling:4}
Applying the algorithm to the ideal~$I$ of Example~\ref{ex:doublestirling} with 
respect to~$k$, we obtain a coupled system of size $5\times14$ when $A$ and $B$
are assumed to have total degree 2 and~1, respectively.
This system has no solution. 
For total degrees 3 and~2, respectively, the obtained system is of size $14\times28$ and has the nontrivial solution
 \begin{alignat*}1
   A &= S_m + S_l + (2 + l + m)S_lS_m - S_lS_mS_n,\\
   B &= \frac{k(k+1)}{k^2-1-n-kn}S_l+\frac{(m+1)k}{k-n-1}S_mS_l,
 \end{alignat*}
hence $A + (S_k-1)B \in I$, and $A\in T_k(I)$. 
\end{example}

\section{Final Comments}

We have seen that algorithms for special functions doing computations in 
Ore algebras are not restricted to special functions with $\partial$-finite 
annihilating ideals. Instead, both
closure properties algorithms and algorithms for definite integration and
summation can be formulated for arbitrary ideals. The 
treatment could be extended further by including, for instance, ideals
of Laurent Ore algebras. 

Our generalized algorithms rely on the notion of ideal dimension as well as
on the notion of polynomial growth we introduced in Definition~\ref{def:p}.
According to this definition, the polynomial growth depends on the monomial
order imposed on the underlying algebra. Future research will focus on 
reducing the notion of polynomial growth to an intrinsic property of the ideal
at hand, as well as to devising an algorithm for computing the polynomial
growth. 
\bibliographystyle{plain}

\end{document}